\documentclass[a4paper]{article}
\usepackage{amssymb}
\usepackage{makeidx}
\usepackage{amsfonts}
\usepackage{amsmath}
\usepackage{amssymb}
\usepackage{geometry}
\usepackage[T1]{fontenc}
\usepackage{lmodern}
\usepackage{amsthm}

\newtheorem{theorem}{Theorem}[section]

\newtheorem{lemma}[theorem]{Lemma}

\theoremstyle{definition}

\newtheorem{remark}[theorem]{Remark}

\newcommand{\R}{\mathbb{R}}
\newcommand{\supp}{\mathop{\textup{supp}}}
\newcommand{\singsupp}{\mathop{\textup{singsupp}}}

\newcommand{\WF}{\mathop{\textup{WF}}}
\newcommand{\Char}{\mathop{\textup{Char}}}

\begin{document}

\title{Microlocal analysis of fractional wave equations}
\author{
 G\"{u}nther H\"{o}rmann\thanks{Faculty of Mathematics, University of Vienna, Oskar-Morgenstern-Platz 1, 1090 Wien, Austria,
 guenther.hoermann@univie.ac.at},
 Ljubica Oparnica\thanks{Faculty of Education in Sombor, University of Novi Sad, Podgori\v{c}ka 4, 25000 Sombor, Serbia,
 ljubica.oparnica@gmail.com},
 Du\v{s}an Zorica\thanks{Mathematical Institute, Serbian Academy of Arts and Sciences, Kneza Mihaila 36, 11000 Belgrade, Serbia, dusan\textunderscore zorica@mi.sanu.ac.rs
 and Department of Physics, Faculty of Sciences, University of Novi Sad, Trg D. Obradovi\'{c}a 4, 21000 Novi Sad, Serbia}}

\maketitle

\begin{abstract}
We determine the wave front sets of solutions to two special cases
of the Cauchy problem for the space-time fractional Zener wave
equation, one being fractional in space, the other being
fractional in time. For the case of the space fractional wave
equation, we show that no spatial propagation of singularities
occurs. For the time fractional Zener wave equation, we show an
analogue of non-characteristic regularity.

\textbf{Key words}: wave front set, space-time fractional wave
equation, Cauchy problem, fractional Zener model, fractional
strain measure
\end{abstract}

\section{Introduction}

This paper is devoted to the microlocal analysis of the solution
to the generalized Cauchy problem for the space-time fractional
Zener wave equation
\begin{equation} \label{Duhamel} Zu\left(
x,t\right) =\partial _{t}^{2}u(x,t)-L_{t}^{\alpha }\partial _{x}
\mathcal{E}_{x}^{\beta }u(x,t)=u_{0}(x)\otimes \delta ^{\prime
}\left( t\right) +v_{0}(x)\otimes \delta \left( t\right)
\end{equation}
considered as an equation on all $\R^2$ with $\supp(u) \subseteq
\{(x,t) \in \R^2 \mid t \geq 0\}$ and $u_0, v_0 \in
\mathcal{E}'(\R)$. The generalized Cauchy problem \eqref{Duhamel}
is derived and analyzed in \cite{AJOPZ}, where existence and
uniqueness of distributional solutions has been shown. In the
present paper we study the wave front set for special cases of
\eqref{Duhamel} when $\alpha=0$, or $\beta=1$.\\

The operators $L_{t}^{\alpha }$ and $\mathcal{E}_{x}^{\beta }$ are
of convolution type,  with respect to one variable only, denoted
by $\ast_t$ and $\ast_x$, respectively, and act on a distribution
$w=w(x,t)$ in the following way
\begin{eqnarray}
L_{t}^{\alpha } w &=&\mathcal{F}_{\tau \rightarrow t}^{-1}\left[
\frac{1+b\, \mathrm{e}^{\mathrm{i}\frac{\alpha \pi }{2}}\left(
\tau -\mathrm{i}0\right) ^{\alpha
}}{1+a\,\mathrm{e}^{\mathrm{i}\frac{\alpha \pi }{2}}\left( \tau -
\mathrm{i}0\right)^{\alpha }}\right] \ast _{t} w,\quad \alpha \in
\left[0,1\right), 0 < a < b, \label{L} \\
\mathcal{E}_{x}^{\beta } w &=& \mathcal{F}_{\xi \rightarrow
x}^{-1}\left[ i \sin(\frac{\beta \pi}{2})\, \text{sgn}(\xi)\,
|\xi|^\beta \right] \ast_{x} w, \quad \beta \in (0,1). \label{E}
\end{eqnarray}
Note that $\mathcal{E}_{x}^{1} w = \frac{\partial }{\partial x} w$
and, in case $0 < \beta < 1$, we have $$\mathcal{E}_{x}^{\beta } w
= \frac{1}{2\Gamma \left( 1-\beta \right) }\frac{1}{{{\left\vert
x\right\vert }^{\beta }}} \ast_{x}\frac{\partial }{\partial x}
w.$$

The operator $L_t^\alpha$, considered as a convolution operator in
one variable, is linear and bounded $L^p(\R) \to L^p(\R)$, $1 < p
< \infty$, by H\"ormander's multiplier theorem (cf.\
\cite[Corollary 8.11]{duoa} or \cite[Theorem 7.9.5]{hermander1}),
since $l_\alpha$, defined by
\begin{equation}
l_\alpha(\tau) = \frac{1+b\, \mathrm{e}^{\mathrm{i}\frac{\alpha
\pi }{2}}\left( \tau -\mathrm{i}0\right) ^{\alpha
}}{1+a\,\mathrm{e}^{\mathrm{i}\frac{\alpha \pi }{2}}\left( \tau -
\mathrm{i}0\right)^{\alpha }}
\end{equation}
is in $L^\infty(\R) \cap C^1(\R \setminus \{0\})$ with derivative
bounded by a constant times $|\tau|^{-1}$. Note that in
\cite{AJOPZ}, the operator $L_t^\alpha$ is employed in its Laplace
transform variant $L_{t}^{\alpha } w =\mathcal{L}_{s\rightarrow
t}^{-1}\left[ \frac{1+bs^{\alpha }}{1+as^{\alpha }}\right] \ast
_{t} w$.

The operator $\mathcal{E}_x^\beta$, acting by convolution in one
variable with $h(x) := |x|^{- \beta}$ (apart from a constant),
following a differentiation, is a bounded linear operator
$W^{1,p}(\R) \to L^q(\R)$, $1 < p < q < \infty$ and $\frac{1}{p} +
\beta = \frac{1}{q} + 1$, since by the
Hardy-Littlewood-Sobolev-inequality \cite[Theorem
4.5.3]{hermander1} the map $w \mapsto h \ast w$ is
continuous $L^p(\R) \to L^q(\R)$ in this setting.\\

The space-time fractional Zener wave equation
\begin{equation}
Zu\left(
x,t\right):=\partial_{t}^{2}u(x,t)-L_{t}^{\alpha}\partial_{x}
\mathcal{E}_{x}^{\beta}u(x,t)=0,\;\;x\in \mathbb{R},\;t>0,
\label{fzwe}
\end{equation}
subject to initial conditions
\begin{equation}
u(x,0)=u_{0}(x),\;\;\frac{\partial }{\partial
t}u(x,0)=v_{0}(x),\;\;x\in \mathbb{R},  \label{pu}
\end{equation}
is derived in \cite{AJOPZ} from the system of three equations: The
equation of motion of a (one-dimensional) deformable body, the
constitutive equation and the non-local strain measure. In
dimensionless form, the system of equations reads
\begin{eqnarray}
\frac{\partial }{\partial x}\sigma (x,t) &=& \frac{\partial
^{2}}{\partial t^{2}}u(x,t),  \label{em} \\
\sigma (x,t)+a\, {}_{0}\mathrm{D}_{t}^{\alpha }\sigma (x,t)
&=&\varepsilon (x,t)+b \, {}_{0}\mathrm{D}_{t}^{\alpha
}\varepsilon (x,t),
\label{ce} \\
\varepsilon (x,t) &=&\mathcal{E}_{x}^{\beta }u\left( x,t\right),
\label{sm}
\end{eqnarray}
where $\alpha \in \left[0,1\right)$, $\beta \in \left(
0,1\right)$, ${}_{0}\mathrm{D}_{t}^{\alpha }$ is the fractional
differential operator, defined as follows. Let $t, \gamma \in
\mathbb{R}$ and $H$ denote the Heaviside function. Then one
defines
\begin{equation*} f_{\gamma }(t)=\left\{
\begin{array}{ll}
\displaystyle\frac{t^{\gamma -1}}{\Gamma (\gamma )}H(t), & \gamma >0, \\
\displaystyle\frac{\mathrm{d}^{N}}{\mathrm{d}t^{N}}f_{\gamma
+N}(t), & \gamma \leq 0,\;\gamma +N>0,\;N\in \mathbb{N}
\end{array}
\right.
\end{equation*}
and for $g \in \mathcal{S}^\prime$, with support in the region
$t>0$,
\begin{equation*}
{}_{0}\mathrm{D}_{t}^{\gamma } g =f_{-\gamma }\ast _{t} g
=\frac{\mathrm{d}}{\mathrm{d}t}f_{1-\gamma }\ast _{t} g.
\end{equation*}

Upon Fourier transform we may solve (\ref{ce}) with respect to
$\sigma$ by
\begin{equation}
\sigma (x,t)=\mathcal{F}_{\tau \rightarrow t}^{-1}\left[
\frac{1+b\,\mathrm{e}^{\mathrm{i}\frac{\alpha \pi }{2}}\left( \tau
-\mathrm{i} 0\right) ^{\alpha
}}{1+a\,\mathrm{e}^{\mathrm{i}\frac{\alpha \pi }{2}}\left( \tau
-\mathrm{i}0\right) ^{\alpha }}\right] \ast _{t}\varepsilon (x,t).
\label{sigma}
\end{equation}
Indeed, by \cite[Example 7.1.17]{hermander1},
$\widehat{f_{1-\alpha}} =\mathrm{e}^{-\mathrm{i}\frac{\left(
1-\alpha \right) \pi }{2}}\left( \tau -\mathrm{i} 0\right)^{\alpha
-1}$, implying also $\widehat{ f_{-\alpha }} =
\mathcal{F}(\frac{d}{dt} f_{1-\alpha }) =\mathrm{i}\tau
\,\mathrm{e}^{-\mathrm{i}\frac{\left( 1-\alpha \right) \pi
}{2}}\left( \tau-\mathrm{i}0\right) ^{\alpha
-1}=\mathrm{e}^{\mathrm{i}\frac{\alpha \pi }{2}}\left( \tau
-\mathrm{i}0\right) ^{\alpha}$, hence we find that (\ref{sigma})
solves (\ref{ce}). Finally, inserting this into \eqref{em} and
observing \eqref{sm}, we arrive at Equation \eqref{fzwe}. As in
\cite{AJOPZ}, we study the Cauchy problem (\ref{fzwe}-\ref{pu})
with distributional initial values in the form (\ref{Duhamel}).

%The framework for our analysis are the spaces of distributions:
%$\mathcal{D}^{\prime }(\mathbb{R})$ dual of
%$\mathcal{D}(\mathbb{R})$, the space of smooth and compactly
%supported functions, and $\mathcal{S}^{\prime}(\mathbb{R})$ the
%dual of the Schwartz space $\mathcal{S}(\mathbb{R})$.  The space
%$\mathcal{S}_{+}^{\prime }$ is a subspace of
%$\mathcal{S}^{\prime}(\mathbb{R})$ consisting of elements
%supported by $[0,\infty )$. The space $\mathcal{E}'(\R)$ is a
%subspace of $\mathcal{D}^{\prime }(\mathbb{R})$ as well as of
%$\mathcal{S}^{\prime}(\mathbb{R})$ consisting of distributions
%with a compact support.\\

In Section \ref{SecSWE} we analyze the microlocal properties of
the space-fractional wave equation and in Section \ref{SecZWE} we
address an analogue of the non-characteristic regularity of
solutions to the time-fractional Zener wave equation.

%%%%%%%%%%%%%%%%%%%%%%%%%%%%%%%%%%%%%%%%%%%%%%%%%%%%%%%%%%%%%%%%%%%%%%%%%%%%%
%%%%%%%%%%%%%%%%%%  SECTION 2 %%%%%%%%%%%%%%%%%%%%%%%%%%%%%%%%%%%%%%
%%%%%%%%%%%%%%%%%%%%%%%%%%%%%%%%%%%%%%%%%%%%%%%%%%%%%%%%%%%%%%%%%%%%%%%%%%%
\section{The space-fractional wave equation}\label{SecSWE}

We consider the solution for the special case of problem \eqref{Duhamel} with $u_0, v_0 \in \mathcal{E}'(\R)$ when $\alpha =0$ and $0 < \beta < 1$, which leads to the so-called \emph{space-fractional wave equation}
\begin{equation}
Zu\left( x,t\right) =\partial _{t}^{2}u(x,t)-\partial
_{x}\mathcal{E}_{x}^{\beta }u(x,t)=u_{0}(x)\otimes \delta ^{\prime
}\left( t\right) + v_0(x) \otimes \delta(t). \label{zet}
\end{equation}
With $b_{\beta } := \sqrt{\sin \frac{\beta \pi }{2}}$ we have
 the solution $u$ with $\supp(u) \subseteq \{ t \geq 0\}$ in the form (cf.\ \cite{AJOPZ})
\begin{equation}
u = u_{0}\ast_{x} \underbrace{\mathcal{F}_{\xi \rightarrow x}^{-1}
       \left[ \cos \left( b_{\beta }
            \left\vert \xi \right\vert^{\frac{1+\beta }{2}}t\right) H(t) \right]}_{E_0}
    + v_0 \ast_x \underbrace{\mathcal{F}_{\xi \rightarrow x}^{-1}
       \left[ \frac{\sin \left( b_{\beta }
            \left\vert \xi \right\vert^{\frac{1+\beta }{2}}t\right)}{b_{\beta }
                \left\vert \xi \right\vert^{\frac{1+\beta }{2}}} H(t) \right]}_{E_1}.
                \label{u-as}
\end{equation}

\begin{remark} \label{fund-rem}
\begin{trivlist}
\item{(i)} We observe that $E_1$ is a \emph{fundamental solution} of $Z$, i.e., $Z E_1 = \delta$. Furthermore, the equations $E_0 = \partial_t E_1$ and $Z E_0 (x,t) = \delta (x) \otimes \delta'(t)$ hold on $\R^2$.

Furthermore, both $E_0$ and $E_1$ are weakly smooth with respect to $t$ when $t \neq 0$. We note that $t \mapsto E_1(t)$ is continuous $\R \to \mathcal{S}'(\R)$ with $E_1(0) = 0$, whereas $\lim_{t \to 0+} E_0(t) = \delta \neq 0 = \lim_{t \to 0-} E_0(t)$; $E_0$ is weakly measurable with respect to $t \in \R$.

\item{(ii)} It is apparent from \eqref{u-as} and the assumption
$u_0, v_0 \in \mathcal{E}'(\R)$ that the partial $x$-Fourier
transform $\mathcal{F}_{x \to \xi}(u)$ of $u$ is a continuous
function with respect to $\xi$ and of moderate growth. Hence the
multiplication $\widehat{g_\gamma} \cdot \mathcal{F}_{x \to \xi}
(u)$ with $\widehat{g_\gamma}(\xi) := |\xi|^\gamma$  ($-1 < \gamma
< 1$) gives a locally integrable function of moderate growth with
respect to $\xi$, and $G_\gamma u := \mathcal{F}^{-1}_{\xi \to
x}(\widehat{g_\gamma} \cdot \mathcal{F}_{x \to \xi} (u))$ is
defined in $\mathcal{S}'(\R^2)$. The same is true, if in place of
$u$ we consider $\tilde{u} = \mathcal{F}^{-1}_{\xi \to x}(\exp(i
b_\beta |\xi|^{\frac{1 + \beta}{2}} t)) \ast_x u_0$ and other
similarly constructed distributions, e.g., $w=G_\sigma\tilde{u}$.
We will repeatedly make use of this fact within the current
section in course of the following proofs.

\item{(iii)} For fixed $t > 0$, the linear operators of
convolution with $E_0(t)$ and $E_1(t)$ are bounded $L^p(\R) \to
L^p(\R)$, if $1 < p < \infty$, by H\"ormander's multiplier
theorem, since their Fourier transforms are in $L^\infty(\R) \cap
C^1(\R \setminus \{0\})$ with derivative bounded by a constant
times $|\xi|^{-1}$ (cf.\ \cite[Corollary 8.11]{duoa} or
\cite[Theorem 7.9.5]{hermander1}).
\end{trivlist}
\end{remark}

\begin{lemma}\label{fund-res} For $j=0,1$, let $E_j^+$ denote the restriction of $E_j$  to the open half-space $\{ t > 0\}$. Then the wave front sets are given by
$$
   \WF(E_0^+) = \WF(E_1^+) =
      \{ (0,t;\xi,0) \mid t > 0, \xi \neq 0\}=:W_0.
$$
\end{lemma}

\begin{proof}

From $\partial_t E_1^+ = E_0^+$ we immediately deduce
\begin{equation}
   \WF (E_0^+) \subseteq \WF(E_1^+) \subseteq
      \WF(E_0^+) \cup \{ (x,t ; \xi,0) \mid t > 0, \xi \neq 0 \}, \tag{$*$}
\end{equation}
where the right-most set corresponds to the characteristic set of $\partial_t$ when considered as partial differential operator on $\R \times \,]0,\infty[$.

\paragraph{Claim 1:}  Both, $\WF(E_0^+)$ and $\WF(E_1^+)$, are contained in $W_0$.

Let $t > 0$, put $e_0^+(\xi,t) := \cos ( b_{\beta } |\xi|^{\frac{1+\beta }{2}} t )$ and $e_1^+(\xi,t) := \sin ( b_{\beta } |\xi|^{\frac{1+\beta }{2}}t)/
(b_\beta |\xi|^{\frac{1+\beta }{2}})$, and choose $\rho \in C^{\infty}(\R)$ such that $\rho(\xi) = 0$ for $|\xi| \leq 1/2$ and $\rho (\xi) =1$ for $|\xi| \geq 1$.  Then at fixed arbitrary $t > 0$ and $j \in \{0,1\}$ we have
$$
E_j^+(t) = \mathcal{F}_{\xi \to x}^{-1} \left(e_j^+(\xi,t)\right) = \mathcal{F}^{-1}_{\xi \to x} \left(e_j^+(\xi,t) \rho(\xi)\right) + \mathcal{F}_{\xi \to x}^{-1} \left(e_j^+(\xi,t) (1 - \rho(\xi))\right) =: F_{j,1}(t) + F_{j,2}(t).
$$
We observe that $(\xi,t) \mapsto e_j^+(\xi,t) (1 - \rho(\xi))$ is continuous, has compact $\xi$-support, and is smooth with respect to $t$, more precisely, $e_j^+ (1 - \rho) \in C^\infty(]0,\infty[, C_{\text{c}}(\R))$, hence by linearity, commutativity with $\frac{d}{dt}$, and continuity of the inverse Fourier transform with respect to $\xi$, we have $F_{j,2} \in C^\infty \left( ]0,\infty[, \mathcal{F}^{-1}\left(C_{\text{c}}(\R)\right) \right) \subseteq C^\infty \left( ]0,\infty[, C^\infty(\R) \right) \subseteq C^\infty(]0,\infty[ \, \times \R)$, thus, $\WF(E_j^+) = \WF(F_{j,1})$.

Note that $a_j(x,t,\xi) := e_j^+(\xi,t) \rho(\xi)$ define functions in $C^\infty(\R^2\times \R)$ ($j=0,1$) and that $a_1(x,t,\xi) = \int_0^t a_0(x,s,\xi) \, ds$. Furthermore, $a_0$ is a symbol of class $S^0_{\frac{1-\beta}{2},\frac{1+\beta}{2}}(\R^2\times \R)$, since $e_0^+$ is the real part of a function of a special case of the type discussed in \cite[Example 1.1.5.]{hermander-paper} (with appropriate choices of parameters and variable names); here, the condition $0 \leq \beta < 1$ is crucial. In fact, the corresponding symbol estimates need only be carried out in the region $|\xi| > 1$ and are elementary in our case. We also see directly that $x$-derivatives of $a_1$ vanish, any $\xi$-derivative of $a_1$ and $a_0$---as well as $a_1$ and $a_0$ themselves---have essentially the same bounds when $(x,t)$ vary in a compact subset and $|\xi| > 1$; furthermore, any $t$-derivative of $a_1$ brings us back to estimating $a_0$ with one order less in the $t$-derivatives, thus $a_1$ is contained in the same symbol class as $a_0$.

To complete the proof of Claim 1, we observe that $F_{j,1}$ ($j=0,1$), being an inverse Fourier transform, can be written as oscillatory integral (in the sense of \cite[Theorem 7.8.2]{hermander1}) with symbol $a_j(x,t,\xi)/(2\pi)$ and phase funtion $\phi(x,t,\xi) = x \, \xi$ in both cases. Thus, according to \cite[Theorem 8.1.9]{hermander1}, the only contributions to the wave front sets can stem from points with stationary phase, i.e.,
\begin{eqnarray*}
WF\left( E_j^+ \right) &\subseteq &\left\{ \left( x,t; \partial_x \phi(x,t,\xi), \partial_t \phi(x,t,\xi) \right) \mid t > 0, \exists \, \xi \neq 0 \colon \partial
_{\xi }\phi \left( x,t,\xi \right) =0\right\}  \\
&=&\left\{ \left( 0,t,\xi ,0\right) \mid t > 0, \xi \neq 0\right\} = W_{0}.
\end{eqnarray*}

\paragraph{Claim 2:} $W_0 \subseteq \WF(E_0^+)$.

Note that due to the symmetry of $\mathcal{F}(E_0^+)$ with respect
to $\xi \mapsto -\xi$ and the result of Claim 1 we have
$(0,t;\xi,0) \in \WF(E_0^+)$ $\Leftrightarrow$ $(0,t;-\xi,0) \in
\WF(E_0^+)$ $\Leftrightarrow$ $(0,t) \in \singsupp(E_0^+)$. Thus,
it suffices to show that $E_0^+$ is nonsmooth along the half axis
$x = 0$, $t > 0$.

We introduce $\widetilde{E}(t)  := \mathcal{F}_{\xi \rightarrow
x}^{-1}
       \left[ \exp\left(i b_{\beta } t |\xi|^{\frac{1 + \beta}{2}}\right) \right]$ and observe that by the elementary relations $\cos(z) = (\exp(iz) + \overline{\exp(iz)})/2$ and $\mathcal{F}^{-1}(\overline{v})(x) = \overline{\mathcal{F}^{-1}(v)(-x)}$ and employing the ad-hoc notation $R^*$ for the pull-back by $R(x,t) = (-x,t)$ we may write
$$
  E_0^+ = \frac{1}{2} \left( \widetilde{E} + \overline{R^* \widetilde{E}} \right).
$$
Since we are here concerned solely with the question of smoothness at the points $(0,t)$ for any $t > 0$ and $t \mapsto E_0^+(t)$ as well as $t \mapsto \widetilde{E}(t)$ is smooth, we may note: $E_0^+$ is non-smooth at $(0,t)$ $\Leftrightarrow$ $E_0^+(t)$ is non-smooth at $x = 0$ $\Leftrightarrow$ $\widetilde{E}(t)$ is non-smooth at $x = 0$.

Let $f:=\mathcal{F}_{\xi \rightarrow x}^{-1}\left[
\mathrm{e}^{\mathrm{i} b_{\beta }\left\vert \xi \right\vert
^{\sigma }}\right] \in \mathcal{S}'(\R)$, abbreviating $\sigma = (1 + \beta)/2$, and observe that for $t > 0$, rescaling by $\xi \mapsto \xi\, t^{1/\sigma}$ on the Fourier transformed side, we have
$\widetilde{E}\left(t\right) =f\left( ./{t^{1/\sigma }}\right)/{t^{1/\sigma }}$ and therefore may reduce the question of smoothness of $\tilde{E}\left( t\right)$ at $x = 0$ further to that of smoothness of $f$ at zero.

Choose $\eta \in
C^{\infty }\left(\mathbb{R}\right)$ such that $\eta =0$ near $\xi
=0$, $\eta =1$ for $\left\vert \xi \right\vert >1$, and write
$\hat{f}=\left( 1-\eta \right) \hat{f}+\eta \hat{f}.$ Since
$\left( 1-\eta \right) \hat{f}\in C_{c}\left( \mathbb{R} \right)$, smoothness of $f$ is equivalent to that of $\mathcal{F}^{-1} \left(\eta \hat{f}\right)$.

Let $\theta>0$ and put
\begin{equation*}
Q_{\theta} w := \mathcal{F}_{\xi \rightarrow x}^{-1}\left[ \frac{\eta
\left( \xi \right) }{\left\vert \xi \right\vert ^{\theta}}\right] \ast
w\;\;\text{and}\;\;m_{\theta}\left( \xi \right)
:=\mathcal{F}_{x\rightarrow \xi }\left[ Q_{\theta}f \right] \left( \xi
\right) =\frac{\eta \left( \xi \right) }{\left\vert \xi
\right\vert ^{\theta}}\mathrm{e}^{\mathrm{i}b_{\beta }\left\vert \xi
\right\vert ^{\sigma }}.
\end{equation*}
Here, $Q_\theta$ is a pseudodifferential operator with
symbol $q_{\theta} (x,\xi) = \eta(\xi)/ |\xi|^{\theta}$, which clearly satisfies $\left\vert q_{\theta}\left( \xi \right) \right\vert = 1 / {\left\vert \xi \right\vert ^{\theta}}$, if $\left\vert \xi \right\vert >1$, hence $Q_\theta$ is elliptic (of order $-\theta$). Thus, smoothness of $f$ at $0$ turns out to be equivalent to smoothness of $Q_{\theta} f$ at $0$. Note that $Q_\theta f$ is smooth off $0$ by essentially the same arguments used as with $E_0$ and the symbol $a_0$ in Claim 1.

The non-smoothness of $Q_\theta f$ at $0$ is shown thanks to an asymptotic result by G. H. Hardy mentioned in \cite[p.\ 357, 5.3(ii)]{stein}, with parameters $a$, $b$ there to be identified with $\sigma$,  $\theta$ respectively; note that $1 > \sigma = (1+\beta)/2 > 1/2$ and let us suppose that $\theta > (1-\frac{\sigma}{2}) / (1-\sigma) > 0$; then we conclude that there is some constant $c_2 > 0$ such that
\begin{equation*}
\left( Q_{\theta}f\right) \left( x\right)
=c_{1}\mathrm{e}^{\mathrm{i}\frac{c_{2}}{\left\vert x\right\vert
^{\alpha }}}\frac{1}{\left\vert x\right\vert ^{\gamma }}+O\left(
\frac{1}{\left\vert x\right\vert ^{\gamma -\frac{\alpha
}{2}}}\right)
\qquad (x \to 0),
\end{equation*}
where $\alpha =\frac{\sigma }{1-\sigma }$ and $\gamma
=\frac{\theta \left(1 - \sigma \right) - 1 + \frac{\sigma
}{2}}{2\sigma - 1}>0.$ Thus, $Q_{\theta} f (x)$ cannot be bounded as $x\rightarrow 0$ and therefore cannot be continuous near $x=0$, which completes the proof of Claim 2.

\medskip

From Claims 1 and 2 in conjunction with the first inclusion relation in ($*$) established at the beginning of the proof, we obtain
$$
  W_0 \subseteq \WF(E_0^+) \subseteq \WF(E_1^+) \subseteq W_0,
$$
which implies that equality holds throughout and completes the proof.
\end{proof}

\bigskip

Based on the results of Lemma \ref{fund-res} we will investigate
the influence of the singularities in the initial data $u_0$ and
$v_0$ on the wave front set of the solution $u$ to (\ref{zet}). We
emphasize that the proof of Theorem \ref{december} below uses only
the inclusion relation $\WF(E_j^+) \subseteq W_0$ in its first
part and provides an independent, more general, proof of equality
in this relation---thus substituting the argument of Claim 2 above
drawing on Hardy's asymptotics by advanced microlocal techniques.

\begin{remark} If $v_0 = 0$ and $t > 0$ the solution formula \eqref{u-as} simply means  $u^+(t) = E_0^+(t) \ast u_0$.  Since $\singsupp(E_0^+(t)) = \{0\}$ a smooth cut-off $\rho$ near $x = 0$ implies $(E_0^+(t)(1 - \rho)) \ast u_0 \in C^\infty(\R)$, hence it suffices to investigate the singularity structure of $(E_0^+(t) \rho) \ast u_0$, where now both convolution factors belong to $\mathcal{E}'(\R)$. At fixed $t$, this enables us to employ the methods and results from \cite[Section 16.3]{hermander2} on singular supports of convolutions (or to extend these techniques to wave front sets as suggested by H\"ormander directly after the statement of \cite[Definition 16.3.2]{hermander2}). Though a bit technical, it is not difficult to see that one will then obtain equality of the closed convex hulls of $\singsupp(u^+(t))$ and $\singsupp(u_0)$. However, even having information on $\WF(u^+(t))$ for every $t > 0$ would not yield precise microlocal information on $\WF(u^+)$ in terms of the two-dimensional directions in the cotangent fiber.
\end{remark}

\begin{theorem} \label{december}Let $u_{0}, v_0\in \mathcal{E}^{\prime
}\left(\mathbb{R}\right)$ and denote by $u^+$ the restriction of the solution $u$ to \eqref{zet} to the half-space of future time $\R \times \, ]0,\infty[$, then $\WF(u^+)$ is invariant under translations $(x,t) \mapsto (x,t+s)$ with $s > 0$ and
\begin{equation*}
\WF\left( u^+ \right) \subseteq \left\{ \left( x,t ; \xi ,0\right) \mid t > 0,  \left(
x,\xi \right) \in WF\left( u_{0}\right) \text{ or }
  \left( x,\xi \right) \in WF\left( v_{0}\right) \right\}.
\end{equation*}
Moreover, in case $v_0$ is smooth we have the more precise statement
$$
  \WF\left( u^+ \right) = \left\{ \left( x,t ; \xi ,0\right) \mid t > 0,  \left(x,\xi \right) \in WF\left( u_{0}\right) \right\},
$$
and similarly $\WF\left( u^+ \right) = \left\{ \left( x,t ; \xi ,0\right) \mid t > 0, \left(x,\xi \right) \in WF\left( v_{0}\right) \right\}$, if $u_0$ is smooth.
\end{theorem}

To prepare for the proof of the theorem we need a technical lemma on ``symbol corrections''.
\begin{lemma}
\label{simbol-y}Let $\sigma \in \left( 0,1\right)$ and $y(\xi,\tau) = - \tau + b_\beta |\xi|^\sigma$.  Let $\Gamma \subseteq \R^2$ (representing the $(\xi,\tau)$-plane) be the union of a closed disc around $(0,0)$ and a closed narrow cone containg the $\tau$-axis and being symmetric with respect to both axes. Let $\Gamma'$ be a closed set of the same shape as $\Gamma$, but with slightly larger disc and opening angle of the cone. Let $\tilde{b}\in S^{0}\left( \mathbb{R}^{2}\times \mathbb{R}^{2}\right)$ such that $\tilde{b}(x,t,\xi,\tau)$ is real, constant with respect to $(x,t)$, homogenous of degree $0$ with respect to $(\xi,\tau)$  away from the disc contained in $\Gamma'$, and such that $\tilde{b}(x,t,\xi,\tau) = 0$, if $(\xi,\tau) \in \Gamma$, $\tilde{b}(x,t,\xi,\tau) = 1$, if $(\xi,\tau) \not\in \Gamma'$. Then $y \, \tilde{b}$ is a symbol belonging to the class $S^{1}\left( \mathbb{R}^{2}\times
\mathbb{R}^{2}\right)$.
\end{lemma}

\begin{proof} By construction of $\tilde{b}$, it suffices to check the symbol
estimates when $(\xi,\tau)$ is outside $\Gamma$, say $|\tau| < c
|\xi|$, and $|\xi| + |\tau|$ is large. The upper bound in order
zero is clear from $|y \, \tilde{b}| \leq C_0 (|\xi|^\sigma +
|\tau|) \leq C_0' (1 + |\xi| + |\tau|)$.  A derivative
$\partial_\xi^{\alpha_1} \partial_\tau^{\alpha_2} (y \,
\tilde{b})$ with $\alpha_1 + \alpha_2 = n \geq 1$ involves (apart
from combinatorial constants) only nonzero terms of the form
$\partial_\tau  y \, \partial_{\xi }^{l}\partial _{\tau }^{m-1}
\tilde{b} = - \partial_{\xi }^{l}\partial _{\tau }^{m-1}
\tilde{b}$ with $l + m = n$ or  $\partial _{\xi }^{k}  y \,
\partial_{\xi }^{l}\partial _{\tau }^{m} \tilde{b}$  with $k + l +
m = n$, hence it suffices to estimate these. We have
$|\partial_{\xi }^{l}\partial _{\tau }^{m-1} \tilde{b}(\xi,\tau)|
\leq C_{l,m} \, (1 + |\xi| + |\tau|)^{0 - l - (m-1)} = C_{l,m} \,
(1 + |\xi| + |\tau|)^{1 - n}$ and
\begin{eqnarray*}
\left\vert \partial _{\xi }^{k}y\left( \xi ,\tau \right) \partial
_{\xi }^{l}\partial _{\tau }^{m} \tilde{b} \left( \xi ,\tau \right)
\right\vert &=& C'\, \left\vert \xi \right\vert ^{\sigma
-k}\left\vert \partial _{\xi }^{l}\partial _{\tau }^{m} \tilde{b} \left( \xi
,\tau \right) \right\vert \leq \tilde{C}\, \frac{\left( 1+\left\vert
\xi \right\vert +\left\vert \tau \right\vert
\right) ^{-l-m}}{1+\left\vert \xi \right\vert ^{k-\sigma }} \\
&\leq &\tilde{C}\, \frac{\left( 1+\left\vert \xi \right\vert
+\left\vert \tau \right\vert \right)
^{-l-m}}{1+\frac{1}{2}\left\vert \xi \right\vert
^{k-1}+\frac{1}{2}\left\vert \xi \right\vert ^{k-1}}\leq
\tilde{C}\, \frac{\left( 1+\left\vert \xi \right\vert +\left\vert
\tau \right\vert \right) ^{-l-m}}{1+\frac{1}{2}\left\vert \xi
\right\vert ^{k-1}+\frac{1}{2c}\left\vert \tau
\right\vert ^{k-1}} \\
&\leq & C \, \left( 1+\left\vert \xi \right\vert +\left\vert
\tau \right\vert \right) ^{1-k-l-m}
=  C \,\left( 1+\left\vert \xi \right\vert +\left\vert
\tau \right\vert \right) ^{1-n}.
\end{eqnarray*}
\end{proof}

\begin{proof}[Proof of the Theorem] We consider the case when $v_0$ is smooth, which we may
immediately reduce to $v_0 = 0$, since the contribution of $E_1^+ \ast_x v_0$ to the solution is smooth.
Put $K=f^{\ast }E_0^+$, where $f \colon \mathbb{R} \, \times ]0,\infty[\, \times \R \to
\mathbb{R} \, \times ]0,\infty[ \, =: \Omega$ is given by $f\left(x,t,y\right) = \left(
x-y,t\right),$ and $f^{\ast }$ is the distributional pull-back in the
sense of \cite[Theorem 6.1.2]{hermander1}. Then $K\in
\mathcal{D}^{\prime }\left(\Omega \times \R
\right)$ and \cite[Theorem 8.2.4]{hermander1} and Lemma \ref{fund-res} imply
\begin{equation*}
\WF\left( K\right) \subseteq \left\{ \left( x,t,y;\xi ,0,-\xi \right)
\mid x=y, \xi \neq 0\right\} =\left\{ \left( x,t,x\right) \mid x,t\in
\mathbb{R} \right\} \times \left\{ \left( \xi ,0,-\xi \right) \mid
\xi \neq 0\right\}.
\end{equation*}

We have $u^+ = u \!\mid_{\{t > 0\}} = u_0 \ast_x E_0^+$, whose
action on test functions $\varphi \in \mathcal{D}(\Omega)$ can be
written in the form $\langle u^+ , \varphi \rangle = \langle K,
\varphi \otimes u_0 \rangle$, if $u_0 \in \mathcal{D}(\R)$, i.e.,
$u_0 \mapsto u^+$ is the linear map $S \colon \mathcal{D}(\R) \to
\mathcal{D}'(\Omega)$ with distribution kernel $K$. Since
$\WF'(K)_{\R} = \{ (y,\eta) \mid \exists (x,t) \colon
(x,t,y;0,0,-\eta) \in \WF(K) \} = \emptyset$, \cite[Theorem
8.2.13]{hermander1} implies that $S$ may be extended to a map
$\mathcal{E}^{\prime}\left(\mathbb{R} \right) \to
\mathcal{D}'(\Omega)$ and satisfies
\begin{equation*}
\WF\left( Su_{0}\right) \subseteq \WF\left( K\right) _{\Omega}\cup
{\WF}' \left( K\right) \circ \WF\left( u_{0}\right),
\end{equation*}
where $\WF\left( K\right) _{\Omega} =\left\{ \left( x,t ; \xi
,\tau \right) \mid  \left( x,t,y ; \xi ,\tau ,0\right) \in
\WF\left( K\right)  \text{ for some } y\in \R\right\} =\emptyset$
and ${\WF}' \left( K\right) =\left\{ \left( x,t,y ; \xi ,\tau
,\eta \right) \mid \left( x,t,y ; \xi ,\tau ,-\eta \right) \in
WF\left( K\right) \right\} =\left\{ \left( x,t,x ; \xi ,0,\xi
\right) \mid t > 0, \xi \neq 0 \right\}$. Thus, we obtain
\begin{eqnarray*}
\WF\left( u^+ \right) &\subseteq &\left\{ \left( x,t ; \xi ,\tau \right) \mid
\exists \left( y,\eta \right) \in \WF\left( u_{0}\right) \colon \left(
x,t,y ; \xi ,\tau ,\eta \right) \in {\WF}^{\prime }\left( K\right) \right\} \\
&\subseteq &\left\{ \left( x,t ; \xi ,\tau \right) \mid t > 0, \exists \left( y,\eta
\right) \in WF\left( u_{0}\right) \colon y=x,\tau =0,\eta =\xi \right\} ,
\end{eqnarray*}
i.e.,
\begin{equation}
\WF\left( u^+ \right) \subseteq \left\{ \left( x,t ; \xi ,0\right) \mid
( x,\xi ) \in \WF ( u_{0}), t > 0 \right\},
\label{u-wo}
\end{equation}
and the remaining part of the proof is concerned with showing that equality holds in (\ref{u-wo}).

As in the proof of Lemma \ref{fund-res} let $\tilde{E}(t)  := \mathcal{F}_{\xi \rightarrow x}^{-1} \left[ \exp\left(i b_{\beta } t |\xi|^{\sigma}\right) \right]$ with $\sigma := (1 + \beta)/{2}$, but this time for any $t \in \R$, and put $\tilde{u}(t) := \tilde{E}(t) \ast u_0$. We have
$ D_t \widehat{\tilde{E}(t)} =
  \frac{1}{i}\partial _{t} \widehat{\tilde{E} (t)} =  b_{\beta }
   \vert \xi \vert ^{\sigma } \mathrm{e}^{\mathrm{i}b_{\beta } \vert
\xi \vert ^{\sigma }t} = b_{\beta }\vert \xi \vert
^{\sigma }\widehat{\tilde{E}(t) }$,
which implies
\begin{equation*}
 Y \tilde{E} :=  -D_{t}\tilde{E} +A_{x}^{\sigma }\tilde{E} =0,
 \quad \tilde{E}\left( 0\right) =\delta,
\end{equation*}
where $A_{x}^{\sigma }w= \mathcal{F}_{\xi \rightarrow
x}^{-1}\left[ b_{\beta }\left\vert \xi \right\vert ^{\sigma
}\right] \ast_x w$ (with $w$ of the type as in Remark
\ref{fund-rem}(ii)). Moreover, $\tilde{u}$ solves the initial
value problem
\begin{equation}
 Y \tilde{u} =
 \left( -D_{t}+A_{x}^{\sigma }\right) \tilde{u}
=\left( -D_{t} \tilde{E} +A_{x}^{\sigma} \tilde{E} \right) \ast_x u_0 = 0,
  \quad \tilde{u}(0) =u_{0}. \tag{$*$}
\end{equation}

Note that, since $b_{\beta}^{2} \left\vert \xi \right\vert ^{2\sigma} = \sin(\beta \pi/2) |\xi|^{\beta + 1}$ is precisely the ``symbol'' of $- \mathcal{E}_x^\beta \, \partial_x$, we have a (commutative) factorization of $Z$ by $\left( D_{t}+A_{x}^{\sigma }\right) \left( -D_{t}+A_{x}^{\sigma}\right) v=-D_{t}^{2}v+A_{x}^{\sigma }A_{x}^{\sigma }v=\partial
_{t}^{2}v+\mathcal{F}_{\xi \rightarrow x}^{-1} (b_{\beta}^{2}\left\vert \xi \right\vert ^{2\sigma }) \ast_x v=Zv$, i.e., $Z = \bar{Y} \cdot Y$, where we have put $\bar{Y} := D_{t}+A_{x}^{\sigma }$.

Before entering the detailed microlocal analysis of $\tilde{u}$
let us anticipate its relevance for $u^+$: We will show in
Equation \eqref{u-tilda-u} below, that in the region with $t > 0$,
$\tilde{u}$ provides a ``lower bound'' for the wave front set of
$u^+$, in fact, we will show equality of the wave front sets at
the end of the proof.

In studying the propagation of singularities for problem ($*$) we
encounter the nuisance that $y\left( \xi ,\tau \right) =-\tau
+b_{\beta }\left\vert \xi \right\vert ^{\sigma }$ is not quite a
symbol of order $1$ in $\xi$ and $\tau$, since $y$ is nonsmooth at
zero and,  furthermore, the symbol estimates obviously fail, e.g.,
for $\vert \partial _{\xi }^{2} y (\xi ,\tau) \vert = \vert \sigma
(\sigma - 1) \vert \xi \vert ^{\sigma -2}\vert$ when $\tau \to
\infty$, there would have to be a bound of decrease $(1 + |\tau| +
|\xi|)^{-1}$ for large $|\tau| + |\xi|$. A remedy of this second
kind of ``symbol failure'' is discussed in \cite[Theorem
18.1.35]{hermander3}, see also a comment below \cite[Theorem
23.1.4]{hermander3}, which we essentially follow in studying the
propagation of singularities for $\tilde{u}$ considered as
solution to
\begin{equation*}
 Y B\, \tilde{u} = B Y \tilde{u}=0, \quad \tilde{u}\left( 0\right) =u_{0},
\end{equation*}%
where $B=\text{op}\left( \tilde{b}\right) \in \Psi ^{0}\left( \mathbb{R}^{2}\right)$ is the pseudodifferential operator associated with a symbol $\tilde{b}$ given as in Lemma \ref{simbol-y}. Thus, $Y B = \text{op}\left(
y \tilde{b}\right) \in \Psi ^{1}\left(\mathbb{R}^{2}\right)$ has principal symbol
\begin{equation*}
q\left( \xi ,\tau \right) :=-\tau\, \tilde{b} \left( \xi ,\tau \right),
\end{equation*}
which is real and homogeneous of degree one, and, modulo a regularizing contribution, can be considered properly supported.

By \cite[Theorem 26.1.1]{hermander4} $\WF(\tilde{u})$ is invariant under the flow corresponding to the Hamiltonian vector field
\begin{equation*}
H_{q}( x,t,\xi ,\tau ) =
\begin{pmatrix}
-\partial _{\xi }q \\
-\partial _{\tau }q \\
\partial _{x}q \\
\partial _{t}q
\end{pmatrix}
 =
\begin{pmatrix}
\tau \partial _{\xi } \tilde{b} \\
\tilde{b} + \tau \partial _{\tau } \tilde{b} \\
0 \\
0
\end{pmatrix}
\end{equation*}
and is contained in the characteristic set $\Char\left( YB\right)
= \mathbb{R}^{2}\times \left\{ \left( \xi ,0\right) \mid \xi \neq
0\right\}$, i.e., $\WF(\tilde{u})\subseteq \Char\left( YB\right)$
In fact, a refinement of the latter inclusion relation is
available, since $\tilde{u} = \tilde{E} \ast_x u_0$ and we may
argue very similar to proof of (\ref{u-wo}), noting (as in Claim 2
of the proof of Lemma \ref{fund-res}) that $\tilde{E}$ is
microlocally equivalent to $\mathcal{F}_{\xi \to x}(\cos(b_\beta
|\xi|^\sigma t))$, and deduce
\begin{equation}
  \WF( \tilde{u}) \subseteq \{ ( x,t ; \xi ,0) \mid ( x,\xi ) \in \WF ( u_{0}), t \in \R \}.
\label{tilde-u-wo}
\end{equation}

Solving the Hamiltonian equations
\begin{eqnarray*}
&&\dot{x}=\tau \,\partial _{\xi } \tilde{b}( \xi ,\tau )
,\;\;\dot{t} = \tilde{b}( \xi ,\tau ) +\tau\, \partial _{\tau
} \tilde{b}( \xi ,\tau )
,\;\;\dot{\xi}=0,\;\;\dot{\tau}=0,\\
&& \text{with } ( x( 0) ,t( 0) ,\xi ( 0)
,\tau ( 0) ) =( x_{0},t_{0},\xi _{0},0) \in \Char( YB),
\end{eqnarray*}
we obtain $\forall s \in \R$: $x( s) =x_{0}$, $t( s) =t_{0} + s \,
\tilde{b}( \xi _{0},0)$, $\xi (s) = \xi _{0}$, $\tau( s) =0$. We
may suppose that $\tilde{b}(\xi_0,0) = 1$, since $\xi_0 \neq 0$
and characteristic sets as well as wave front sets are conic with
respect to the cotangent fibers. Hence the bicharacteristic flow
evolves along the $t$-direction with fixed cotangent directions
$(\xi_0,0)$ only. Therefore we have $(x_0,t_0; \xi_0,0) \in
\WF(\tilde{u})$ if and only if $(x_0,0; \xi_0,0) \in
\WF(\tilde{u})$. We claim that the latter is in turn equivalent to
$(x_0,\xi_0) \in \WF(u_0)$, from which, together with
(\ref{tilde-u-wo}), we may then conclude
\begin{equation}
  \WF( \tilde{u}) = \{ (x,t;\xi ,0) \mid ( x,\xi ) \in \WF( u_{0}), t \in \R \}.
\label{prop-sing}
\end{equation}
We have claimed: $(x_0,0;\xi_0,0) \in \WF( \tilde{u})$ $\Leftrightarrow$ $( x_{0},\xi _{0}) \in WF(u_{0})$.\\
The implication `$\Rightarrow$' follows from \eqref{tilde-u-wo}. For the converse, note that, according to \cite[Theorem 8.2.4]{hermander1} and the fact that $\WF(\tilde{u}) \subseteq \Char(Y B)$ contains no directions $(0,\tau)$ in the fiber, we may write $\tilde{u}(t) = f_t^* \tilde{u}$ for any $t \in \R$, where $f_t(x) = (x,t)$ as a  map $\R \to \R^2$, and obtain
$$
   \WF(\tilde{u}(t)) \subseteq f_t^* \WF(\tilde{u}) =
   \{ (x,\xi) \mid (x,t;\xi,0) \in \WF(\tilde{u}) \} =
    \{ (x,\xi) \mid (x,0;\xi,0) \in \WF(\tilde{u}) \},
$$
where the last equality follows from the Hamiltonian invariance
proven above. In particular, when $t = 0$ we have $\tilde{u}(0) =
u_0$, so that $\WF(u_0) = \WF(\tilde{u}(0)) \subseteq \{ (x,\xi)
\mid (x,0;\xi,0) \in \WF(\tilde{u}) \}$, which proves the part
`$\Leftarrow$' of the claim and thus establishes
\eqref{prop-sing}.

%%%%%%%% new variant to improve the WF-relation between U^+ and \tilde{u} %%

We are now ready to clarify the microlocal relation between $\tilde{u}$ and $u^+$: In the subdomain with $t > 0$ we have
\begin{equation} \label{YA}
  \bar{Y} u^+ =\left( D_{t} E_0^+ +A_{x}^{\sigma }E_0^+ \right) \ast_x u_{0} =
 A_{x}^{\sigma}\tilde{E} \ast_x u_{0} =
  A_{x}^{\sigma}\tilde{u},
\end{equation}
since
$\mathcal{F}_{x\rightarrow \xi }\left[ D_{t}E_0^+\left( t\right) +A_{x}^{\sigma
}E_0^+\left( t\right) \right] = - i \partial _{t}\left( \cos
\left( b_{\beta }\left\vert \xi \right\vert ^{\sigma }t\right) \right)
+b_{\beta }\left\vert \xi \right\vert ^{\sigma }\cos \left( b_{\beta
}\left\vert \xi \right\vert ^{\sigma }t\right) =b_{\beta }\left\vert \xi
\right\vert ^{\sigma }\widehat{\tilde{E}\left( t\right)}$.

Denoting by $\tilde{u}^+$ the restriction of $\tilde{u}$ to the half-plane of positive time we claim that the following two inclusions hold:
$$
  (\text{I})\enspace \WF(\bar{Y} u^+) \subseteq \WF(u^+), \quad \text{and} \quad
  (\text{II})\enspace \WF(\tilde{u}^+) \subseteq \WF(A_x^\sigma \tilde{u}^+).
$$
Since $\bar{Y} = D_t + A_x^\sigma$ and $D_t$ clearly is a microlocal, i.e, $\WF(D_t w) \subseteq \WF(w)$ for any $w \in \mathcal{D}'(\R^2)$, we may reduce (I) to the statement $\WF(A_x^\sigma u^+) \subseteq \WF(u^+)$. Furthermore, $A_x^\sigma$, acting only in the $x$-variable, commutes with restriction to $t > 0$, hence in intermediate steps we may  consider $A_x^\sigma$ as convolution on $\R^2$ with $\mathcal{F}^{-1}_{\xi \to x}(b_\beta |\xi|^\sigma) \otimes \delta(t)$ and restrict to $t > 0$ afterwards.

Note that we have $\tilde{u} = B_x^\sigma A_x^\sigma \tilde{u} =
A_x^\sigma B_x^\sigma \tilde{u}$, where $B_x^\sigma$ is
$x$-convolution with the inverse Fourier transform of the locally
integrable function $\xi \mapsto 1/(b_\beta |\xi|^\sigma)$. Thus,
the statements (I) and (II) are equivalent to showing
$\WF(G_\sigma u^+) \subseteq \WF(u^+)$ and $\WF(G_{-\sigma}
A_x^\sigma\tilde{u}^+) \subseteq \WF(A_x^\sigma\tilde{u}^+)$  with
$G_\gamma$ and $g_\gamma$ specified as in Remark
\ref{fund-rem}(ii) with $\gamma = \sigma$ or $\gamma = -\sigma$
(both in the range $-1 < \gamma < 1$); clearly, $G_\gamma$ also
commutes with restriction to $t > 0$ and $A_x^\sigma = b_\beta
G_{\sigma}$, $B_x^\sigma = G_{-\sigma} / b_\beta$.

The operator $G_\gamma$ can be considered as convolution on $\R^2$
with the distribution $g_\gamma (x) \otimes \delta(t)$, where
$\widehat{g_\gamma}(\xi) = |\xi|^\gamma$. The one-dimensional
homogeneous distribution $g_\gamma$ can be determined explicitly
via \cite[Example 7.1.17]{hermander1}, showing directly that
$\singsupp(g_\gamma) = \{ 0 \}$, and we easily deduce from
\cite[Theorem 8.1.8.]{hermander1} that $\WF(g_\gamma) = \{ (0,\xi)
\mid \xi \neq 0\}$. Hence \cite[Theorem 8.2.9]{hermander1} gives
$$
  \WF(g_\gamma \otimes \delta) \subseteq
  \{ (0,0;\xi,\tau) \mid (\xi,\tau) \neq (0,0) \} \cup
  \{ (x,0;0,\tau) \mid x \in \R, \tau \neq 0 \}
$$
and, recalling from Remark \ref{fund-rem}(ii) that $G_\gamma w =
(g_\gamma \otimes \delta) \ast w$ is defined in case of $w = u^+$
or $w = A_{x}^{\sigma}\tilde{u}^+$, we may prove by cut-off
techniques the appropriate extension of \cite[Equation
(8.2.16)]{hermander1} to these cases and obtain
$$
     \WF(G_\gamma w) \subseteq \WF(w) \cup
     \underbrace{\{ (x,t;0,\tau) \mid \exists y \in \R \colon (y,t;0,\tau)
       \in \WF(w) \}}_{=: \WF_{\text{vert}} (w)}.
$$
Equations \eqref{prop-sing} and \eqref{u-wo} show
$\WF_{\text{vert}} (\tilde{u}^+) = \emptyset$ and
$\WF_{\text{vert}} (u^+) = \emptyset$, respectively, hence the
proof of (I) and (II) is complete.

We may now put (I) and (II) to use with the outermost equalities in \eqref{YA} and arrive at the following:
\begin{equation}
\WF (\tilde{u}^+) \subseteq \WF ( A_{x}^{\sigma}\tilde{u}^+) =
\WF (\bar{Y} u^+) \subseteq \WF (u^+). \label{u-tilda-u}
\end{equation}

%%% end of new variant to improve the WF-relation between U^+ and \tilde{u} %%%%

In summary, combining Equations \eqref{prop-sing} and \eqref{u-tilda-u} with \eqref{u-wo} we obtain
\begin{multline*}
  \{ (x,t;\xi ,0) \mid t > 0, ( x,\xi ) \in \WF( u_{0}) \} =
    \WF(\tilde{u}^+) \subseteq \WF (u^+)\\  \subseteq
    \{ (x,t;\xi ,0) \mid t > 0, ( x,\xi ) \in \WF( u_{0}) \},
\end{multline*}
hence we have, in fact, equality in all places of the above chain of inclusions, thereby the proof of the theorem in case $v_0 = 0$ is completed.

As shown in Lemma \ref{fund-res} the microlocal structure of
$E_1^+$ is equivalent to that of $E_0^+$, hence we have the same
kind of wave front set statement with $v_0$ in place of $u_0$, if
$u_0$ is smooth, since in this case $u^+ = E_1^+ \ast_x v_0$ plus
a smooth contribution steming from $u_0$.

Finally, the solution in the general case $u_0, v_0 \in \mathcal{E}'(\R)$ is just the sum of the two solutions for the special cases $v_0 = 0$ and $u_0 = 0$, hence its wave front set is contained in the corresponding union. Invariance of the wave front set under positive time translations follows in this case as well, since it was established via the operator factorization $Z = \bar{Y} \cdot Y$ with subsequent ``symbol correction factor'' $B$ and is valid for solutions $w$ of $Y B w = 0$ independent of initial values.
\end{proof}

\begin{remark} The result on the wave front set of $u^+$ in the above theorem implies, in particular, smoothness of $u^+$ considered as a map from time into distributions on space  (cf.\   \cite[(23.65.5)]{diodone}), i.e., $u^+ \in C^{\infty }( ]0,\infty[,\mathcal{D}^{\prime }(\mathbb{R}))$; in addition, we have $u^+(t) \in \mathcal{S}'(\R)$ for every $t > 0$.
\end{remark}

%%%%%%%%%%%%%%%%%%%%%%%%%%%%%%%%%%%%%%%%%%%%%%%%%%%%%%%%%%%%%%%%%%%%%%%
%%%%%%%%%%%%%% SECTION 3 %%%%%%%%%%%%%%%%%%%%%%%%%%%%%%%%%%%%%%%%%%%
%%%%%%%%%%%%%%%%%%%%%%%%%%%%%%%%%%%%%%%%%%%%%%%%%%%%%%%%%%%%%%%%%%%%%%
\section{The time-fractional Zener wave equation}\label{SecZWE}

For the special case of (\ref{Duhamel}) when $\beta =1$ and $0
\leq \alpha < 1$ we obtain the so-called \emph{time-fractional
Zener wave equation}
\begin{equation}
Zu\left( x,t\right) =\partial _{t}^{2}u(x,t)-L_{t}^{\alpha
}\partial _{x}^{2}u(x,t)=u_{0}(x)\otimes \delta ^{\prime }\left(
t\right) + v_0(x) \otimes \delta(t), \label{zet-koz}
\end{equation}
whose unique solvability by distributions supported in a forward
cone has been established in \cite{KOZ10}. Here we show a kind of
non-characteristic regularity of the solution $u$ to problem
\eqref{zet-koz}.

The ``Fourier symbol'' of $Z$ is $z(\xi,\tau) = - \tau^2 + l_\alpha(\tau) \xi^2$ with
$$
  l_\alpha(\tau) :=  \frac{1+b\,{e}^
        {i\frac{\alpha \pi }{2}}(\tau -i 0)^{\alpha}}{1+a\,{e}^{i\frac{\alpha
                                                  \pi }{2}}(\tau - i 0)^{\alpha }}
   = \frac{1 + b\,i^{\alpha}\, \text{sgn}(\tau)\,\vert\tau\vert ^{\alpha}}{ 1 +
     a\, {i}^{\alpha}\,\text{sgn}(\tau)\, \vert\tau\vert^{\alpha }},
$$
to which we apply a conic cut-off to obtain a smooth symbol in both variables $(\xi,\tau)$, similarly as in Lemma \ref{simbol-y} above.

\begin{lemma}
\label{l-alfa-je-simbol}
Let $\Gamma \subseteq \R^2$ (representing the $(\xi,\tau)$-plane) be the union of a closed disc around $(0,0)$ and a closed narrow cone containg the $\xi$-axis and being symmetric with respect to both axes. Let $\Gamma'$ be a closed set of the same shape as $\Gamma$, but with slightly larger disc and opening angle of the cone. Let $\tilde{b}\in S^{0}\left( \mathbb{R}^{2}\times \mathbb{R}^{2}\right)$ such that $\tilde{b}(x,t,\xi,\tau)$ is real, constant with respect to $(x,t)$, homogenous of degree $0$ with respect to $(\xi,\tau)$  away from the disc contained in $\Gamma'$, and such that $\tilde{b}(x,t,\xi,\tau) = 0$, if $(\xi,\tau) \in \Gamma$, $\tilde{b}(x,t,\xi,\tau) = 1$, if $(\xi,\tau) \not\in \Gamma'$. Then $p := \tilde{b}\, z$ is a symbol belonging to the class $S^{2}\left( \mathbb{R}^{2}\times
\mathbb{R}^{2}\right)$.
\end{lemma}

The proof is a variation of that of Lemma \ref{simbol-y}.

\begin{theorem}
For the wave front set of $u^+$, the restriction of the solution $u$ to (\ref{zet-koz}) to forward time $t > 0$, we have the inclusion
\begin{equation*}
  \WF ( u^+ ) \subseteq \{ ( x,t;\xi,\tau) \mid x \in \R, t > 0,
       \xi \neq 0, \tau^{2} = \frac{b}{a}\, \xi^2 \textup{ or } \tau = 0\}.
\end{equation*}
\end{theorem}

%\todo{more precise statement about $(x,t)$-singularities, i.e.,
%$\singsupp(u^+)$?} We emphasize that we consider the cotangent
%directions with $\tau = 0$ an artefact of the method of proof
%rather than a true part of the wave front set. \todo{can we find
%an argument to get rid of the artefact?}

\begin{proof}
Let $B$ and $P$ be the pseudo-differential operators associated
with the symbols $\tilde{b}$ and $p$, respectively, constructed in
Lemma \ref{l-alfa-je-simbol} according to arbitrary, but fixed,
$\Gamma$ and $\Gamma'$ chosen as specified there. We have $P = B
Z$ and therefore
\begin{equation*}
   P u^+ = B Z u^+ = 0.
\end{equation*}
By non-characteristic regularity \cite[Theorem 18.1.28]{hermander3},
\begin{equation*}
 \WF\left( u\right) \subseteq \Char\left( P \right),
\end{equation*}
where the characteristic set is $\Char\left( P \right) =
\left( \mathbb{R}^{2}\times \left(\mathbb{R}^{2}\setminus
  \left\{ \left( 0,0\right) \right\} \right) \right) \setminus M$ with $M$ being defined as the set of all points $\left( x_{0},t_{0},\xi _{0},\tau _{0}\right) $ such that there exist $c>0,$ $R>0$ and a conic neighborhood $V$ of $(
\xi _{0},\tau _{0}) $ such that the following estimate holds:
\begin{equation}\label{estim-p}
   \forall (\xi ,\tau) \in V,\;\xi ^{2}+\tau ^{2} \geq R^2 \colon \quad
\left\vert p\left( \xi ,\tau \right) \right\vert \geq
  c\,(\xi^2 + \tau^2).
\end{equation}

\begin{enumerate}
\item We have $\mathbb{R}^{2}\times \left( \Gamma \setminus\{
(0,0)\} \right) \cap M = \emptyset$, since
$\tilde{b}(\xi_0,\tau_0) = 0$ whenever $(\xi_0,\tau_0) \in
\Gamma$. As $\Gamma$ gets more and more narrow (and smaller around
the origin) only points of the form $(x_0,t_0,\xi_0,0)$ will
remain with this property.

\item We have no definite  decay properties of the symbol in all
of $\R^2 \times \left(\Gamma' \setminus \Gamma\right)$, but this
will not be required as we let later shrink both $\Gamma' \supset
\Gamma$ to $\R \times \{0\}$, causing $\Gamma' \setminus \Gamma
\to \emptyset$.

\item Suppose $( x_{0},t_{0},\xi _{0},\tau _{0}) \in
\mathbb{R}^{2} \times \left( \R^2 \setminus \Gamma' \right)$, which will leave only points with $\tau_0 \neq 0$ upon the shrinking process of $\Gamma'$ and $\Gamma$.

\begin{enumerate}
\item If $\tau _{0}^{2}=\frac{b}{a}\xi _{0}^{2}$, then the estimate
(\ref{estim-p}) must fail in any conic neighborhood of $(\xi_0,\tau_0)$, since for $\lambda > 0$ we have
$$
  \left\vert p( \lambda \xi_0 , \lambda \tau_0 ) \right\vert
   =\left\vert z ( \lambda \xi _{0},\lambda \tau _{0}) \right\vert  =
    \lambda^2\,  \xi_0^2 \underbrace{\left\vert  - \frac{b}{a}  +
        \frac{b\, i^\alpha\, \text{sgn}(\tau_0) |\tau_0|^\alpha + \lambda^{-\alpha}}
        {a\, i^\alpha\, \text{sgn}(\tau_0) |\tau_0|^\alpha + \lambda^{-\alpha}}
          \right\vert}_{=: d(\lambda)},
$$
where $d(\lambda) \to 0$ as $\lambda \to \infty$, which makes a lower bound of the form $\left\vert p( \lambda \xi_0 , \lambda \tau_0 ) \right\vert \geq c\, \lambda^2$ to hold for all $\lambda \geq R / \sqrt{\xi_0^2 + \tau_0^2}$ impossible. Thus, $( x_{0},t_{0},\xi _{0},\tau _{0}) \not\in M$.

\item If $\tau_{0}^{2}\neq \frac{b}{a}\xi _{0}^{2}$, we define a closed
conic neighborhood of the point $\left( \xi _{0},\tau _{0}\right)$
by $V:=\left\{ \left( \xi ,\tau \right) \in \mathbb{R}^{2}\mid
\left\vert \frac{\tau ^{2}-\frac{b}{a}\xi ^{2}}{\xi ^{2}+\tau
^{2}}\right\vert \geq c_{0}-\delta \right\},$ where
$c_{0}:=\left\vert \frac{\tau _{0}^{2}-\frac{b}{a}\xi
_{0}^{2}}{\xi _{0}^{2}+\tau _{0}^{2}} \right\vert$ and $0 < \delta < c_0$. Let $V_R := V \cap \{(\xi,\tau) \mid \xi^2 + \tau^2 \geq R^2 \}$ and suppose that $R > 0$ is large enough and $\delta$ chosen sufficiently small to ensure $V_R \cap \Gamma' =\emptyset$ as well as $V_R \cap \left\{ \left( \xi ,\tau \right) \mid
\tau ^{2}=\frac{b}{a}\xi ^{2}\right\} =\emptyset$.

Let $(\xi,\tau) \in V_R$, then $\tau^2 \geq (c_0 - \delta) (\xi^2 - \tau^2) + \frac{b}{a} \xi^2 \geq \min(c_0 - \delta, \frac{b}{a}) (\xi^2 + \tau^2) \geq c_1 R^2$. Since $l_\alpha(\tau) \to \frac{b}{a}$ ($|\tau| \to \infty$) we may thus choose $R$ large enough to have $l_\alpha(\tau) - \frac{b}{a} < \frac{c_0 - \delta}{2}$, if $\xi^2 + \tau^2 \geq R^2$ and $(\xi,\tau) \in V$. Thus, $(\xi,\tau) \in V_R$ implies
\begin{multline*}
   |p(\xi,\tau)| = |z(\xi,\tau)| = |\tau^2 - l_\alpha(\tau) \xi^2 | =
    | \tau^2 - \frac{b}{a} \xi^2 - (l_\alpha(\tau) + \frac{b}{a}) \xi^2 | \\
    \geq | \tau^2 - \frac{b}{a} \xi^2| - |l_\alpha(\tau) + \frac{b}{a}|\, \xi^2
    \geq (c_0 - \delta) (\xi^2 + \tau^2) - \frac{c_0 - \delta}{2} \xi^2
    \geq \frac{c_0 - \delta}{2} (\xi^2 + \tau^2).
\end{multline*}
Therefore we have in this case $( x_{0},t_{0},\xi _{0},\tau _{0}) \in M$.

\end{enumerate}
\end{enumerate}

To summarize,

\begin{equation*}
  \WF(u^+) \subseteq \Char(P) \subseteq  \mathbb{R}^{2}\times \left(
       \left( \Gamma \setminus\{ (0,0)\} \right)
        \cup \left( \Gamma' \setminus \Gamma \right)
        \cup \{ (\xi_0,\tau_0) \mid \tau_0^2 = \frac{b}{a}\xi_0^2\}
        \right).
\end{equation*}
This result holds for any $\Gamma$ and $\Gamma'$ chosen arbitrarily according to the specifications in the previous lemma. Letting $\Gamma' \supseteq \Gamma$ both shrink toward the $\xi$-axis yields the claim of the theorem, since we may use the intersection of all corresponding $(\Gamma,\Gamma')$-dependent sets in the middle and in the right part of the above chain of inclusions.
\end{proof}

\section*{Acknowledgement}

This research is supported by project P25326 of the Austrian
Science Fund and by projects $174005$, $174024$ of the Serbian
Ministry of Education and Science, as well as by project
$114-451-947$ of the Secretariat for Science of Vojvodina.

%\bibliographystyle{plain}
%\bibliography{dz}

\end{document}